\documentclass[10pt]{amsart}
\usepackage{amssymb}
\usepackage{amsmath}
\usepackage[foot]{amsaddr}
\usepackage{url}
\title{Kolmogorov Complexity and The Garden of Eden Theorem}
\author{Andrey Alpeev}\thanks{This research is supported by the Chebyshev Laboratory  (Department of Mathematics and Mechanics, St. Petersburg State University)  under RF Government grant 11.G34.31.0026}
\address{Chebyshev Laboratory, St. Petersburg State University, 14th Line, 29b, Saint Petersburg, 199178 Russia}
\newtheorem{prop}{Proposition}
\newtheorem{deff}{Definition}
\newtheorem{theor}{Theorem}

\newcommand{\dom}{\mathop{dom}}
\newcommand{\cont}{\mathop{cont}}
\begin{document}
\begin{abstract}
Suppose $\tau$ is a cellular automaton over an amenable group and a finite alphabet. Celebrated Garden of Eden theorem states, that pre-injectivity of $\tau$ is equivalent to non-existence of Garden of Eden configuration. In this paper we will prove, that imposing some mild restrictions , we could add another equivalent assertion: non-existence of Garden of Eden configuration is equivalent to preservation of asymptotic Kolmogorov complexity under the action of cellular automaton.  It yields a characterisation of the cellular automata, which preserve the asymptotic Kolmogorov complexity.
\end{abstract}
\maketitle
keywords: Kolmogorov complexity, cellular automata, garden of Eden theorem 
\section{Introduction}

Suppose $G$ is an amenable group and $A$ is a finite set, called an alphabet. 
Suppose we fixed some F\o lner sequence $\lbrace F_n \rbrace$. 
For an infinite element $x \in A^G$ in the discourse of the Kolmogorov complexity it is very natural to consider some kind of mean information along the F\o lner sequence:
\[
 hc(c)=\limsup_{n \to \infty} \frac{C(x{\vert}_{F_n})}{\lvert F_n \rvert}
\]

 In this way we should make some assertions, to guarantee well behaviour of the defined quantity. If we are working with a group like $\mathbb{Z}^d$ we do not care about computability. Passing to more general case we have to explicitly require our group to be computable. Then, there is a choice of an object to measure the complexity: the string, obtaining by listing of letters, marking elements of given $F_n$ in some order, or the whole object: the map from $F_n$ to $A$. In the second case we get an additional information from the configuration of $F_n$. To avoid such troubles, we introduce the notion of the modest F\o lner sequence, that is such F\o lner sequence, that complexities of its elements are asymptotically negligible to their sizes: $C(F_n)=o(\lvert F_n \rvert)$. This will help us establish equivalence of the two definitions of the asymptotic Kolmogorov complexity.

It is not hard to see, that asymptotic Kolmogorov complexity do not increases under the action of cellular automaton(proposition \ref{asymptotic complexity properties}). The point of interest is to describe those cellular automata, which preserves asymptotic complexity in terms of some well-known properties of cellular automata. The first result of such kind, known to author, contained in the paper \cite{CRI}, there proved, that every invertible cellular automaton over $\mathbb{Z}$ with finite alphabet preserves the asymptotic Kolmogorov complexity, which is defined in the following way:
\[
 \limsup_{i\to \infty} \frac{C(x_i)}{2 i + 1}  
\]
there $C(x_i)$ is the Kolmogorov complexity of the restriction of  $x$ to the segment $[-i,i]$.

The Garden of Eden theorem states, that for the cellular automaton with finite alphabet over an amenable group surjectivity(that is non-existence of Garden of Eden configurations) and pre-injectivity are equivalent(cellular automata is called pre-injective, if every two different configurations with same images under action of cellular automaton, differ on infinite set). We will prove in this paper, that the class of cellular automata, which preserve asymptotic Kolmogorov complexity is exactly the class of surjective and pre-injective cellular automata.  

This paper organized in the following way. In section \ref{computable amenable groups} we discuss the notion of a computable group and prove existence of a computable F\o lner sequences in computable groups. In section \ref{kolmogorov complexity} we remind the notion of the Kolmogorov complexity and its basic properties. Also, there is proved a crucial proposition \ref{complexity and size}, which equip us with the lower bounds of the preimages of a computable functions. In section \ref{cellular automata} we remind some definitions, concerning cellular automata. In section \ref{entropy and complexity} we discuss properties of entropy and asymptotic Kolmogorov complexity. At first, we give the definition of the modest F\o lner sequence, Then, we remind the definition of the entropy and define the asymptotic Kolmogorov complexity. After, we prove the basic properties of the asymptotic complexity.
In proposition \ref{asymptotic complexity properties} we prove, that the entropy of effectively closed set bounds the asymptotic complexity of its elements.
After it we remind the Garden of Eden theorem for amenable groups and Curtis-Hedlund-Lyndon theorem. In the end of this section we give the proof of the main result, theorem \ref{main}, which states, that asymptotic Kolmogorov complexity along the modest F \o lner sequence is presserved by the action of the cellular automaton over the computable amenable group if and only if this automaton is preinjective. 

\section{Computable amenable groups}\label{computable amenable groups}

\begin{deff}
 A pair of a countable or finite set and a bijection of it with an enumerable subset of the natural numbers called a constructible set. 
\end{deff}

\begin{deff}
 Let $G$ be some some group structure on $\mathbb{N}$, such that $0$ is the identity element. Then, if function $(x,y)\mapsto x \circ y$ is computable, we will call the group $G$ computable.
\end{deff}

It is obvious, that the operation $x \mapsto x^{-1}$ in computable group is computable function.

Throughout this paper we will assume, that $G$ is computable, if not, we could enhance our computability class by a suitable oracle.

\begin{deff}
Suppose $G$ is a countable group. We call a sequence $\lbrace F_n \rbrace$ of its finite subsets a F\o lner sequence for this group, if for every $g \in G$
\[
\lim\limits_{n\to \infty}\frac{\lvert F_n g\Delta F_n \rvert}{\lvert F_n\rvert} = 0
\]
there $\Delta$ denotes the symmetric difference. 
\end{deff}

\begin{deff}
We call a countable group amenable, if it has a F\o lner sequence.
\end{deff}

\begin{prop}\label{folner reformulated}
Suppose, $G$ is a countable amenable group, $\lbrace F_n \rbrace$ --- its F\o lner sequence. Then for every finite subset $S$ of $G$ we have
\[
\lim\limits_{n\to \infty}\frac{\lvert F_n S\Delta F_n \rvert}{\lvert F_n\rvert} = 0
\]
\end{prop}

\begin{proof}
Since $S=\lbrace s_1, \ldots, s_k\rbrace$, we have
\begin{multline*}
F_n S \Delta F_n  = \bigcup_{i=1}^{k}{(F_n s_i \setminus F_n)} \cup \bigcap_{j=1}^{k}{(F_n \setminus F_n s_j)} \subset \bigcup_{i=1}^{k}{(F_n s_i \Delta F_n)} \cup {(F_n \setminus F_n s_1)} \subset \\ \subset \bigcup_{i=1}^{k}{(F_n s_i \Delta F_n)} \cup {(F_n \Delta F_n s_1)} 
\end{multline*}
From which we could infer, using the definition of F\o lner sequence
\[
\lim_{n\to \infty} \frac{\lvert F_n S\Delta F_n \rvert}{\lvert F_n\rvert} \leq \lim_{n \to \infty}\frac{\sum_{i=1}^{k}{\lvert F_n s_i \Delta F_n \rvert} + \lvert F_n \Delta F_n s_1\rvert}{\lvert F_n \rvert} = 0
\] 
\end{proof}

\begin{prop}
Suppose $G$ is an infinite countable amenable group, and $\lbrace F_n \rbrace$ is a F\o lner sequence on it, then $\lvert F_n\rvert \to \infty$
\end{prop}

\begin{proof}
Suppose, on the contrary, for some constant $c$ we have infinitely many indices $i$, such that $\lvert F_i\rvert \leq c$. Then, passing to subsequence, we will get a F\o lner sequence $\lbrace F_{n_i}\rbrace$ with $\lvert F_{n_k}\rvert \leq c$. Take some finite subset $S$ of $G$ with $\lvert S \rvert \geq 2 c$ we get a contradiction with the proposition \ref{folner reformulated}:
\[
\frac{\lvert F_{n_i} S \Delta F_{n_i} \rvert}{\lvert F_{n_i}\rvert}\geq 1
\]   
for every $i$.
\end{proof}


Every finite subset of the natural numbers could be computably encoded by the natural numbers, so, they forms the constructible set. It is clear, that set-theoretic operations with the finite subsets are computable, and function $(A,B) \mapsto AB$ there $A$ and $B$ are the finite subsets of computable group is computable too. 

\begin{prop}\label{existence of folner}
 There exists a computable F\o lner sequence , with $\lvert F_n\rvert \geq n$ in every computable amenable group.  
\end{prop}
\begin{proof}
 Fix some computable enumeration of the finite subsets. Let $F_n$ be the first set, such that \[\frac{\lvert F_n i\Delta F_n \rvert }{\lvert F_n\rvert}\leq \frac{1}{n} \text{ for every
i, }0\leq i\leq n \] and $\lvert F_n \rvert \geq n$.

${F_n}$ is a desired computable F\o lner sequence.
\end{proof}

\section{Kolmogorov complexity}\label{kolmogorov complexity}

We will use the notation $C(x)$ for the plain complexity and $C(x \vert y)$ for the conditional complexity, for the definition see \cite{LV08}, definition 2.1.2, p. 106 . Plain complexity of the pair $(x,y)$ is denoted as $C(x,y)$, for the definition see \cite{LV08}, example 2.1.5, p. 109. 

\begin{prop}\label{complexity properties}
For the plain Kolmogorov complexity the following statements hold:
\begin{enumerate}
 \item $C(x) \leq \log x + O(1)$
 \item if $f$ is computable, then $C(f(x)) \leq C(x) + O(1)$
 \item There exist a constant $c$, such that $\left\lvert \lbrace x \vert C(x) \leq n\rbrace \right\rvert \leq c 2^n$
 \item $C(x,y) = C(y,x)+ O(1)$
 \item $C(x,y) \leq C(x) + C(y) + 2 \log \min(C(x),C(y)) + O(1)$
 \item $C(x|y) \leq C(x) + O(1)$
 \item $C(x,y) = C(x) + C(y\vert x) + O(\max(\log C(y),\log C(x)))$
\end{enumerate}

\end{prop}

\begin{proof}
For the proofs, see \cite{LV08}:\\
(1) is theorem 2.1.2 on the page 108, \\
(2) and (3) are simple consequences of the definition of plain complexity,\\
(4) is follows from (2) and definition of complexity of pair,\\
(5) is proved in example 2.1.5 on the page 109, \\
(6) is from exercise 2.1.5 on page 113,\\
(7) is a slight reformulation of the theorem 2.8.2 on the page 190.  

\end{proof}

\begin{prop}\label{inverse complexity}
 Suppose, $f$ is a computable function. If $x=f(y)$, then
\[
 \log \lvert f^{-1}(x) \rvert  \geq C(y) - C(x) - O(\log \lvert C(y) - C(x) \rvert)
\]
\end{prop}

\begin{proof}
 Fix some computable enumeration $t(x,i)$ ($i \leq 0$) of the preimage of the function $f$, such that if $t(x,i)$ does not fail, then $t(x,j)$ does not fail for all $0 \leq j\leq i$ and
$t(x,i) \neq t(x,j)$ (if both do not fail) for $i \neq j$ . Then, there exists $i < \lvert f^{-1}(x)\rvert$, such that $y=t(x,i)$. Using proposition \ref{complexity properties}, we
get
\[
 C(y)=C(t(x,i)) \leq C(x,i) + O(1) \leq C(x) + \log \lvert f^{-1}(x) \rvert + \log \log \lvert f^{-1}(x) \rvert + O(1),
\]
so
\[
\log \lvert f^{-1}(x) \rvert + \log \log \lvert f^{-1}(x) \rvert \geq C(y) - C(x) + O(1) 
\]
which implies desired estimate. 
 
\end{proof}

\begin{prop}\label{complexity and size}
 Suppose, that ${\lbrace V_i \rbrace}_{i \in I}$ is an enumerable family of enumerable sets, there $I$ is some constructible set. Suppose, there exists a function $f: \mathbb{N} \to
\mathbb{N}$, such that for every $i \in I$ we have $\lvert V_i\rvert \leq f(i)$, then for every $x\in V_i$ the following holds:
\[
 C(i,x)\leq \log \lvert V_i \rvert + C(i) + O(log C(i)) 
\]

\end{prop}

\begin{proof}
  Follows from proposition \ref{complexity properties}. 
\end{proof}

\section{Cellular automata}\label{cellular automata}
Suppose $G$ is a group and $A$ is a finite set. We will call $A$ an alphabet. We will call $A^G$ a configuration space.
There is a natural left action of $G$ on $A^G$ defined in the following way:
\[
(gx)(g')=x(g^{-1}g').
\]
\begin{deff}
 We will call a word every map from any finite subset of $G$ to $A$.
\end{deff}

The set of all words is denoted by $A^{\star}$. For a word $w$ we denote $\dom w$ its domain. By $\lvert w \rvert$ we will denote the size of its domain. Consider a word $w$, by the
definition it is a map from the finite set $\lbrace g_1, \ldots g_{\lvert w \rvert}\rbrace$ to the alphabet $A$ 

We will endow $A^G$ with the product topology, assuming that $A$ endowed with the discrete topology. By Tychonoff theorem, $A^G$ is a compact Hausdorff topological set. Sets 
\[
 U(w)=\left\lbrace \left. x\in A^G \right \vert x{\vert}_{\dom w}=w \right\rbrace
\]
are clopen and forms the base of the topology on $A^G$.

Assuming $G$ is computable, it is obvious, that every word could be computably encoded by some natural number.

Since all finite subsets of natural numbers and  all words forms constructible sets, it make sense to consider the Kolmogorov complexity of such objects, and we will use the same notation as fo plain complexity: $C(B)$, $C(w)$, there $B$ is a finite subset, $w$ is a word.

Consider a word $w$, by the definition it is a map from the finite set $\lbrace g_1, \ldots g_\rbrace$. Assume, that $g_i < g_j$ for $1\leq i < j \leq \lvert w \rvert $ (remind, that group $G$ is defined on the set $\mathbb{N}$). Then, string $w(g_1)w(g_2)\ldots w(g_{\lvert w \rvert}) $ is called a content of word $w$, and denoted as $\cont w$.  By the definition, for word $w$ holds $C(w)=C(\dom w, \cont w)$.

\begin{deff}
 Cellular automaton over group $G$ and the finite alphabet $A$ is a map $\tau: A^G \to A^G$ such that there exist a finite set $S \subset G$, and a map $\mu: A^S \to A$, for that the following
equality holds:
\[
 \tau(x)(g)=\mu(g^{-1}x{\vert}_S)
\]

Set $S$ is called a memory set, $\mu$ is called a rule.
 
\end{deff}

\begin{deff}
 
Suppose $\tau$ is a cellular automaton with a memory set $S$ and a rule $\mu$, Suppose, $w$ is a word, such that $\dom w = BS$ for some subset $B$ of $G$. Then we will define 
\[
 \tau(w)(g)=\mu(g^{-1}x{\vert}_S)
\]

for $g \in B$. So it is no more, than the restriction of the action of the cellular automaton to the word $w$.
\end{deff}
Evidently, if $G$ is a computable group and $\tau$ is a cellular automaton over it, then restriction of $\tau$ on words is computable function. 

\section{Entropy and asymptotic Kolmogorov complexity}\label{entropy and complexity}

We want our F\o lner sequence not to affect Kolmogorov complexity  
\begin{deff}
 We will call a F\o lner sequence $\lbrace F_n \rbrace$ a modest F\o lner sequence, if 
\[
 \lim\limits_{n\to \infty} \frac{C(F_n)} {\lvert F_n \rvert} = 0 
\]
 
\end{deff}

It obviously follows from the proposition \ref{existence of folner}, that there exist at least one modest F\o lner sequence.

It would be convinient for us to use the following notation for the restriction of configuration $x\in A^G$ to the subset $B$ of G: $x{\vert}_{B}$. If $B$ is a finite set, then this restriction is a word, so, if in addition, $G$ is a computable group, then this would justify the notation $C(x{\vert}_{B})$. We would use the same notation for the restriction of words as well: $w{\vert}_{B}$ if $B$ is a subset of the domain of word $B$.

Let us fix some modest F\o lner sequence $\lbrace F_n \rbrace$.

\begin{deff}
Consider some subset $X$ of $A^{G}$. We will define its entropy as
\[
 h(X)=\limsup\limits_{n \to \infty} \frac{\log \left\lvert{ X {\vert}_{F_n}}\right\rvert}{\lvert F_n\rvert}
\]
  
\end{deff}

\begin{prop}
For the entropy the following holds:
\begin{enumerate}
 \item $h(X) \leq \log \lvert A \rvert$
 \item $h(A^G)=\log \lvert A\rvert$
 \item $h(\tau(X)) \leq h(X)$.
\end{enumerate}

\end{prop}

\begin{proof}
 For the proof see \cite{CAG}, propositions 7.7.2 and 5.7.3 on pages 125,126.
\end{proof}

\begin{deff}
 For the element $x \in A^G$ we will define its asymptotic complexity as
 \[
 hc(x)=\limsup\limits_{n \to \infty} \frac{C( x {\vert}_{F_n})}{\lvert F_n\rvert}
\]
\end{deff}

\begin{prop}
For every $x \in A^G$ we have  
\[
 hc(x)=\limsup\limits_{n \to \infty} \frac{C \left({ \cont (x {\vert}_{F_n}) }\right) }{\lvert F_n\rvert}
\]

\end{prop}

\begin{proof}
It is enough to prove, that
\[
 \lim\limits_{n \to \infty} \frac{C(x{\vert}_{F_n})-C( \cont (x {\vert}_{F_n}))}{\lvert F_n\rvert}=0
\]
 Using proposition \ref{complexity properties} we have
\begin{multline*}
 \lim\limits_{n \to \infty} \frac{\left\lvert C(x{\vert}_{F_n})-C( \cont (x {\vert}_{F_n}))\right\rvert}{\lvert F_n\rvert} = \\ 
= \lim\limits_{n \to \infty} \frac{\left\lvert C(F_n \vert (x{\vert}_{F_n})) + O(\log (\max (C(F_n), C(\cont (x {\vert}_{F_n})))))\right\rvert} {\lvert F_n\rvert} = 0 
\end{multline*}
the last equality is true, since
\[
C(F_n \vert (x{\vert}_{F_n})) \leq C(F_n) + O(1) =o(\lvert F_n \rvert),
\]
(by the proposition \ref{complexity properties} and the definition of modest F\o lner sequence),  and  
\[
 C(\cont (x {\vert}_{F_n})) \leq \lvert F_n \rvert \log \lvert A \rvert  + O(1). 
\]
by the proposition \ref{complexity properties}.

\end{proof}

The following proposition allows us to estimate the difference of the complexities of two words in terms of the symmetric difference of their supports.

\begin{prop}\label{complexity difference}
 Suppose, $w_1$ and $w_2$ are two words, then 
\begin{multline*}
 \left\lvert C(w_1)-C(w_2) \right \rvert \leq  \lvert \dom (w_1) \Delta \dom (w_2) \rvert \cdot \log \lvert A \rvert +\\+ O(\max(C(\cont (w_1), \cont(w_2))))
\end{multline*}

\end{prop}

\begin{proof}
 Evidently, there exists an algorithms, which recovers $w_1$ from the triple $w_2$, $\dom (w_1)$ and $\cont(w_1\vert_{\dom(w_1 \setminus w_2 )})$. From the proposition
\ref{complexity properties} follows
\[
 C(w_1) \leq C(w_2) + O(C(\dom(w_1))) + C(w_1) + \log \lvert A \rvert \cdot \lvert \dom (w_1) \Delta \dom (w_2) \rvert.
\]
By symmetry we get
\[
 C(w_2) \leq C(w_1) + O(C(\dom(w_2))) + C(w_2) + \log \lvert A \rvert \cdot \lvert \dom (w_1) \Delta \dom (w_2) \rvert.
\] 
Combining this two we get the desired estimate.

\end{proof}

Let us endow $A^G$ with the Bernoulli measure $\nu$.

\begin{prop}\label{asymptotic complexity properties}
 The following holds:
\begin{enumerate}
 \item for every $x\in A^G$ we have $hc(x) \leq \log \lvert A \rvert$
 \item for almost every $x \in A^G$ we have $hc(x)=\log \lvert A \rvert$
 \item for every $x \in A^G$ we have $hc(\tau(x)) \leq hc(x)$.
\end{enumerate}

\end{prop}

\begin{proof}
The following inequality implies the first statement :
\[
  C(\cont(x\vert_{F_n})) \leq \log \lvert F_n\rvert\lvert A \rvert + O(1), 
\]

 For the proof of the second statement consider a set 
\[
 E_{\varepsilon}=\lbrace x\in A^G \vert hc(x) < (1-\varepsilon)\log \lvert A \rvert \rbrace
\]
We have
\[
 E_{\varepsilon} = \bigcup_{i>0} \bigcap_{j>i} \lbrace x \in A^G \vert C(\cont (x\vert F_j)) \leq (1-\varepsilon) \lvert F_j\rvert \log \lvert A \rvert\rbrace 
\]
But since 
\begin{multline*}
\nu\left(\left\lbrace x \in A^G \vert C(\cont(x\vert F_j)) \leq (1-\varepsilon) \lvert F_j\rvert\log \lvert A \rvert\right\rbrace\right) \leq \\ \leq \frac{\left\lvert \left\lbrace {y \vert C(y) <
(1-\varepsilon)
\lvert F_j\rvert \log \lvert A \rvert} \right\rbrace \right\rvert}{{\lvert A \rvert}^{\lvert F_j \rvert}} \leq  {\lvert A \rvert}^{-\varepsilon \lvert F_j\rvert} \to 0
\end{multline*}
We have $\nu(E_{\varepsilon})=0$. Taking a sequence $\varepsilon_n \to 0$, we see, that $\nu(\cup_{i>0} {E_{\varepsilon_i}})=0$, so, almost every $x \in A^G$ has $hc(x)=\log \lvert A
\rvert $, because $hc(x) \leq \lvert A \rvert$ for every $x$.




 
For the proof of the third statement, consider
\[
 hc(x)=\limsup\limits_{n \to \infty} \frac{C( x {\vert}_{F_n})}{\lvert F_n\rvert}
\]
and 

\[
 hc(\tau(x))=\limsup\limits_{n \to \infty} \frac{C( \tau(x) {\vert}_{F_n})}{\lvert F_n\rvert}
\]
but

\[
 \tau(x) {\vert}_{F_n} = \tau(x {\vert}_{F_n S}),
\]
so

\[
 C(\tau(x) {\vert}_{F_n}) \leq C(x {\vert}_{F_n S})+O(1).
\]
This mean, it is enough to prove, that
\[
 \lim\limits_{n\to\infty}\frac{C(x {\vert}_{F_n S})-C(x {\vert}_{F_n})}{\lvert F_n \rvert }=0 
\]
Proposition \ref{complexity difference} implies
\begin{equation*}
 \left\lvert C(x {\vert}_{F_n S})-C(x {\vert}_{F_n}) \right\rvert \leq \lvert F_n S \Delta F_n \rvert \cdot \log \lvert A \rvert + O(\max(C(F_n),C(F_n S)))
\end{equation*}
By the definition of the F\o lner sequence, 
\[
 \lvert F_n S \Delta F_n \rvert = o(\lvert F_n \rvert)
\]
Using proposition \ref{complexity properties}, since $F_n S$ is computable from $F_n$, and by the definition of modest F\o lner sequence we have
\[
 C(F_n S) \leq C(F_n) + O(1) = o(\lvert F_n \rvert), 
\]
which finishes proof.

\end{proof}

\begin{deff}
 An open set $V$ is called an effectively open set, if 
\[
 V=\bigcup\limits_{w \in W} U(w)
\]
for some enumerable set $W$ of words.

A closed set is called effectively closed, if it is a complement of an effectively open set.

\end{deff}

\begin{prop}\label{entropy and asymptotic complexity}
 If $X$ is an effectively closed set, then for every $x \in X$ we have 
\[
 hc(x) \leq h(X).
\]

\end{prop}

\begin{proof}
It is clear, that $\lvert X {\vert}_{F_n}\rvert \leq {\lvert A \rvert} ^ {\lvert F_n \rvert}$

Let us fix some rational $q>h(X)$. By the definition of the entropy, there exists a constant $c$, such that
\[
 \lvert X {\vert}_{F_n}\rvert \leq c 2^{q \lvert F_n \rvert}
\]
If there exists an enumerable family of enumerable sets $\lbrace V_f \rbrace$, there $f \in
\lbrace F_1, F_2, \ldots \rbrace$,  such that $X {\vert}_{f} \subset V_f$ and $\lvert V_{F_n}\rvert
\leq c 2^{ q \lvert F_n \rvert}$ for
all $n$, then by the proposition \ref{complexity and size} we would have for $y \in V_n$ 
\[
 C(y, F_n) \leq O(C(F_n)) + q \lvert F_n \rvert
\]
and therefore, for every $x \in X$
\[
 hc(x)\leq q
\]
But since $q$ is an arbitrary rational, bigger than $h(X)$, we would have $hc(x) \leq hc(X)$ for every $x \in X$.

Let us prove existence of $\lbrace V_f \rbrace$.
$X$ is an effectively closed set, so
\[
 A^G \setminus X = \bigcup_{w \in W} U(w)
\]
for some enumerable set of words $W$.
Consider a word $v$ with $\dom v = f$. Suppose, that $U(v) \cap X = \varnothing$. Then, 
\[
U(v) \subset \bigcup_{w \in W} U(w),
\]
and since it is an open covering of the compact set, we have, $U(v)$ could be covered by some finite subset. This mean, that if $U(v) \cap X = \varnothing$, we could realize it in
a finite time. Therefore, the set of such $v$-s is enumerable(and computably depends on n). Really, assertion, that $U(v)$ is covered by $U(w_1), \ldots, U(w_k)$ is equivalent to the fact, that for every word $t$ with domain 
\[
 E = \dom (v) \bigcup_{i=1}^k{\dom(w_i)}
\]
from the fact, that $\left(t{\vert}_{\dom v}\right)=v$
follows, that there exists $1 \leq i \leq k$, such that $\left(t{\vert}_{\dom w_i}\right)=w_i$
This means, that sets 
\[
 V_{f,k}=(A^G \setminus \bigcup_{i=1}^{k}{U(w_i)}){\vert}_f
\]
forms the computable family of the finite sets, and we could get $V_f$ to be the first $V_{f,k}$ with
$ \lvert V_{f,k} \rvert \leq c \cdot 2^{q \lvert f \rvert}$.      

\end{proof}

\begin{deff}
A cellular automaton $\tau$ over a group $G$ and an alphabet $A$ is called pre-injective, if for every $x,y \in A^G$, such that $x$ and $y$ coincides outside some finite set, and $\tau(x)=\tau(y)$, we have $x=y$.  
\end{deff}

\begin{deff}
A configuration $x\in A^G$ is called a Garden of Eden configuration, if its pre-image is empty.   
\end{deff}

\begin{theor}[The Garden of Eden theorem]
For a cellular automaton $\tau$ over a group $G$ and an alphabet $A$ the following are equivalent:
\begin{enumerate}
 \item There is no Garden of Eden configuration.
 \item $\tau$ is pre-injective.
 \item $h(\tau(A^G))=h(A^G)$.
\end{enumerate}
\end{theor}

\begin{proof}
See \cite{CAG} theorem 5.8.1 on page 128.
\end{proof}

\begin{theor}[Curtis-Hedlund-Lyndon theorem]
Suppose $\tau$ is a map from $A^G$ to itself, there $G$ is a group and $A$ is a finite set. Then the following are equivalent:
\begin{enumerate}
\item $\tau$ is a cellular automaton.
\item $\tau$ is continuous and shift-invariant, that is for every $x \in A^G$ and $g \in G$ we have $\tau(gx)=g \tau(x)$. 
\end{enumerate}
\end{theor}
\begin{proof}
See \cite{CAG}, theorem 1.8.1 on page 20.
\end{proof}

\begin{prop}
 Suppose, that cellular automaton $\tau$ has a Garden of Eden configurations. Then for some $x \in
A^G$ we have $hc(\tau(x)) < hc(x)$.
\end{prop}

\begin{proof}
By the Garden of Eden theorem, if there exists a Garden of Eden configuration, then $h(\tau(A^G)) < h(A^G)$. By the proposition \ref{asymptotic complexity properties} there exist an element $x \in A^G$ with
$hc(x)=\log \lvert A \rvert$. The set $\tau(A^G)$ is a closed subset of the set $A^G$, because it is image of the compact set under the action of the continuous map. 
Since the statement $U(v)\cap \tau(A^G)=\varnothing$ is equivalent to ${\tau}^{-1}(v)=\varnothing$, we have, that the set $\tau(A^G)$ is effectively closed. So, by the proposition \ref{entropy and asymptotic complexity}, we have, that $hc(\tau(x)) \leq h(\tau(A^G)) < \log \lvert A \rvert = hc(x)$. 
\end{proof}

\begin{prop}
 Suppose we have $x \in A^G$, such that $hc(\tau(x)) < hc(x)$. Then the cellular automaton is not pre-injective. 
\end{prop}
\begin{proof}
Without loss of generality we may assume, that $S=S^{-1}$, and that $S$ contains the identity.
Since $hc(\tau(x)) < hc(x)$, there is a subsequence $\lbrace n_k \rbrace$ and two constants $a<b$,
such that 
\[
\frac{C(\tau (x) {\vert}_{F_{n_k}})}{\lvert F_{n_k}\rvert} < a < b
< \frac{C (x {\vert}_{F_{n_k}})}{\lvert F_{n_k}\rvert}  
\]
Using the proposition \ref{complexity difference} we could assume(maybe, passing to subsequence and slightly modifying constants $a$ and $b$), that 
\[
\frac{C(\tau(x){\vert}_{F_{n_k}})}{\lvert F_{n_k}\rvert} < a' < b' < \frac{C( x {\vert}_{F_{n_k}
S})}{\lvert F_{n_k}\rvert}
\]
then, since 
\[
\tau(x {\vert}_{F_{n_k} S})= \tau(x) {\vert}_{F_{n_k}}
\]
we could use the proposition \ref{inverse complexity} and denoting $y=\tau(x)$ we get
\[
\log \lvert {\tau}^{-1}(y{\vert}_{F_{n_k}})\rvert \geq 
C(x{\vert}_{F_n S}) - C(y{\vert}_{F_n}) - O(\log \lvert C(x{\vert}_{F_n S}) - C(y{\vert}_{F_n})  \rvert)
\]
which imply, that for some big enough $N$ and positive constant $c$, for every $k>N$ the following holds:
\[
\log \lvert {\tau}^{-1}(y{\vert}_{F_{n_k}})\rvert \geq c \cdot \lvert F_{n_k} \rvert
\]
so
\[
\lvert {\tau}^{-1}(y{\vert}_{F_{n_k}})\rvert \geq 2^{c \lvert F_{n_k}\rvert}
\]
Consider a set $T$ of pairs of words $(v, w)$, there $f(v)=w$, $\dom (w) = F_{n_k} S S$ and $v{\vert}_{F_{n_K}}=\tau(x){\vert}_{F_{n_k}}$. Its cardinality is at least $2^{c \lvert F_{n_k}\rvert}$ for $k>N$.  
That is, for big enough $k$, by the pigeonhole principle there exists at least two pairs $(v_1,w_2)$ and
$(v_2,w_2)$ in $T$, such that $v_1$ and $v_2$ coincides outside the the set $F_{n_k}$, and $w_1=w_2$.
Really, $w_1$ and $w_2$ coincides on $F_{n_k}$, and we have ${\lvert A
\rvert}^{o(\lvert F_{n_k}\rvert)}$ variants for the filling $F_{n_k} S S \setminus
F_{n_k}$ and $F_{n_k} S \setminus F_{n_k}$ by elements of $A$. 

Extending $w_1$ to some $x_1 \in A^G$ and $w_2$ to some $x_2 \in A^G$ in such a way, that they will coincide outside $F_{n_k}$, we will get, that automaton is not pre-injective, since $\tau(x_1)=\tau(x_2)$.

\end{proof}

Combining last two propositions and the Garden of Eden theorem, we get the following result 
\begin{theor}\label{main}
 Suppose, that $G$ is an infinite computable amenable group, a cellular automaton $\tau$ defined over it, and a modest F\o lner sequence selected(that is, such a F\o lner sequence, that $C(F_n)=o(\lvert F_n \rvert)$). Then the following are equivalent 
\begin{enumerate}
 \item $\tau$ has no Garden of Eden configurations
 \item $\tau$ is pre-injective
 \item $h(\tau(A^G))=h(A^G)$
 \item $hc(\tau(x))=hc(x)$ for every $x \in A^G$.
\end{enumerate}
 
\end{theor}

\end{document}